\documentclass{article}

\usepackage{amsmath,amsfonts,amssymb,amstext,bbm,amsthm}

\usepackage{natbib}

\theoremstyle{plain}
\newtheorem{theorem}{Theorem}

\newtheorem{lemma}{Lemma}

\newtheorem{proposition}[theorem]{Proposition}

\theoremstyle{definition}

\newtheorem{example}{Example}
\newtheorem{definition}{Definition}

\usepackage{float}
\floatstyle{ruled}
\newfloat{algorithm}{tbp}{loa}
\floatname{algorithm}{Algorithm}

\usepackage{multirow}

\usepackage{balance} 

\usepackage{courier}

\begin{document}

\title{On Monotonicity of Number-Partitioning Algorithms}

\author{Erel Segal-Halevi}

\maketitle

\begin{abstract}
An algorithm for number-partitioning is called value-monotone if whenever one of the input numbers increases, the objective function (the largest sum or the smallest sum of a subset in the output) weakly increases. This note proves that the List Scheduling algorithm and the Longest Processing Time algorithm are both value-monotone. This is in contrast to another algorithm --- MultiFit --- which is not value-monotone.
\end{abstract}

\section{Introduction}
In a \emph{number partitioning} problem, the input is a sequence $\mathbf{x}$ of positive numbers and an integer $m\geq 2$.
The goal is to partition the numbers in $\mathbf{x}$ into $m$ subsets with sums that are as nearly equal as possible. Two common objectives are minimizing the largest sum and maximizing the smallest sum.
This problem is known to be NP-hard even for $m=2$, and has many approximation algorithms.

Interestingly, the study of approximation algorithms for number-partitioning began with the study of \emph{anomalies} of algorithms for job scheduling. In 
a simple job-scheduling instance, the input is a set $S$ of \emph{jobs}, each of which has a pre-specified processing-time, and a set $m$ of identical \emph{machines}. The goal is to schedule each job to one of the machines, such that the latest completion time of a machine (also called the \emph{makespan}) is minimized. This problem is clearly equivalent to  number-partitioning with the objective of minimizing the largest sum.
An \emph{anomaly} in a job-scheduling algorithm is a situation in which the jobs become shorter, but the maximum completion time becomes longer.
The first two papers about approximation algorithms for scheduling 
\citep{graham1966bounds,graham1969bounds} were mainly concerned in bounding the size of such anomalies. 

One scheduling algorithm that exhibits anomalies is the algorithm known as \emph{List Scheduling} (LS). This algorithm handles scheduling problems that are more general than number-partitioning in that jobs may have \emph{depenencies} --- some jobs may need the output of earlier jobs in order to start. 
Algorithm LS processes the jobs in a fixed order. For each job $j\in\{1,\ldots,m\}$, if all dependencies of $j$ have finished, LS schedules $j$ to the first machine that becomes available.

\begin{example}[LS algorithm and its anomalies; based on \citet{graham1969bounds}]
\label{exm:ls}
There are $m=3$ machines. There are nine jobs, their processing-times are
\begin{align*}
30,21,22,20,40,40,40,40,90,
\end{align*}
and the four $40$ jobs depend on the $20$ job.

LS schedules $30, 21, 22$ each on a different machine. 
The $20$ is scheduled on machine \#2, since it is becomes available first. 
The four $40$ jobs have to wait until this job finishes, in time $41$. Meanwhile, LS schedules the $90$ on machine \#1 when it becomes available. The final schedule is:
\begin{itemize}
\item Machine \#1: $30 + 90 = 120$
\item Machine \#2: $21 + 20 + 40 + 40 = 121$
\item Machine \#3: $22 + [19 \text{ idle time}] + 40 + 40 =121$
\end{itemize}
so the makespan is $121$.

Now, suppose all jobs become shorter and their processing-times are
\begin{align*}
22,11,12,10,30,30,30,30,80,
\end{align*}
where the four $30$ jobs still depend on the $10$ job.
LS schedules $22, 11, 12$ each on a different machine. 
The $10$ is scheduled on machine \#2. 
It finishes in time $21$, when both machines \#2 and \#3 are available. So the first two $30$ jobs are scheduled to these machines, and the third one is scheduled to machine \#1. 
Now, machines \#2 and \#3 become available before machine \#1, so the last two jobs are assigned to them. The final schedule is:
\begin{itemize}
\item Machine \#1: $22 + 30 = 52$
\item Machine \#2: $11 + 10 + 30 + 30 = 81$
\item Machine \#3: $12 + [9 \text{ idle time}] + 30 + 80 =131$
\end{itemize}
so the makespan is $131$ --- an increase of $10$!
\qed
\end{example}
While dependencies are relevant in job scheduling, they are less relevant for general number partitioning. 
This raises the question of whether LS has similar anomalies when there are no dependencies.%
\footnote{
One can check that, in Example \ref{exm:ls} without dependencies, there are no anomalies --- the initial makespan is 160 and the makespan after the decrease is 131.
} Strangely, I could not find an answer to this question in the large literature on scheduling and number partitioning.
Without dependencies, algorithm LS can be described simply as:
\begin{itemize}
\item Initialize $m$ empty bins.
\item For each number $x_i$ in the input sequence:
\\
put $x_i$ in a bin in which the current sum is smallest
\\
(if there are ties, use the bin with the smallest index).
\end{itemize}
To see that the question is not trivial, consider another example.
\begin{example}[LS algorithm for number-partitioning]
Suppose $m=2$ and the input sequence is $18, 10, 6, 4$.
Then LS returns the partition $(18), (10+6+4)$ with sums $18,20$.

Now, suppose the third input increases and the new sequence is $18, 10, 9, 4$. 
Then LS returns the partition $(18+4, 10+9)$ with sums $22,19$. The second sum decreased! However, both the maximum and the minimum sums increased.
\end{example}

To state the question formally, some definitions are needed.
Given an algorithm $A$, denote by $A(\mathbf{x},m)$ the output of $A$ on inputs $\mathbf{x}$ and $m$. 
Given two sequences $\mathbf{x} = (x_1,\ldots,x_n)$ and $\mathbf{y} = (y_1,\ldots,y_n)$, the notation $\mathbf{x}\geqq \mathbf{y}$ means that $x_i\geq y_i$ for all $i\in[m]$.

\begin{definition}
An algorithm $A$ for number-partitioning is called \emph{value-monotone} if, for any integer $m\geq 2$ and for any two input sequences $\mathbf{x},\mathbf{y}$ such that $\mathbf{x}\geqq \mathbf{y}$, 
the maximum subset sum in $A(\mathbf{x},m)$
is at least 
the maximum subset sum in $A(\mathbf{y},m)$,
and 
the minimum subset sum in $A(\mathbf{x},m)$
is at least 
the minimum subset sum in $A(\mathbf{y},m)$.
\end{definition}

This note proves the following.
\begin{theorem}
\label{thm:mon-ls}
Algorithm LS for number-partitioning (identical-machines scheduling with no job dependencies)
is value-monotone.
\end{theorem}

A variant of algorithm LS, which has superior approximation ratios for both maximizing the largest sum and minimizing the smallest sum, is called \emph{Longest Processing Time first} or LPT:
\begin{itemize}
\item Order the inputs from large to small, such that $x_1\geq \cdots \geq x_n$.
\item Run LS on the resulting sequence.
\end{itemize}
\begin{theorem}
\label{thm:mon-lpt}
Algorithm LPT for number-partitioning is value-monotone.
\end{theorem}

Not all number-partitioning algorithms are value-monotone. 
A third algorithm, called \emph{MultiFit} \citep{coffman1978application}, is known not to be value-monotone even without dependencies. For completeness, an example is presented after the theorems are proved.

\begin{proposition}
Algorithm MultiFit for number-partitioning is not value-monotone.
\end{proposition}

All results are proved in \textbf{Section \ref{sec:value-mon}}.

While searching the literature on scheduling for monotonicity properties, the only one I could find is a different kind of monotonicity, called \emph{speed-monotonicity}, that is relevant when machines may have different speeds. This properety is briefly explained in \textbf{Section \ref{sec:related}}.

\section{Value Monotonicity}
\label{sec:value-mon}
Given an input sequence $\mathbf{x}$ and an integer $m\geq 2$, 
the output of a number-partitioning algorithm $A$ is denoted by $A(\mathbf{x},m)$. It is a sequence of $m$ subsets, $P_1,\ldots,P_m$.
For all $j\in[m]$, denote by $v(P_j)$ the sum of values in $P_j$.

To prove value-monotonicity, a slighly more general property is proved first.

\begin{definition}
A real vector $(s_1,\ldots,s_m)$ \emph{dominates} another vector $(t_1,\ldots,t_m)$
if  there is a permutation $\pi$ of $[m]$ such that, for all $i\in[m]$,
$s_{\pi (i)}\geq t_{i}$.
\end{definition}
Note that, if $\mathbf{s}$ dominates $\mathbf{t}$, then both $\max \mathbf{s} \geq \max \mathbf{t}$ and 
$\min \mathbf{s} \geq \min \mathbf{t}$.
Therefore, to prove that an algorithm is value-monotone, it is sufficient to prove that, when one input increases, 
the new vector $(v(P_1),\ldots,v(P_m))$ dominates the old one.

\subsection{List Scheduling}
Theorem \ref{thm:mon-ls} follows from the following.
\begin{lemma}
\label{lem:mon-ls}
Let $m\geq 2$ be a fixed integer,  $\mathbf{x} := (x_1,\ldots,x_j,\ldots,x_n)$ a sequence of positive numbers,
and $\mathbf{x'} := (x_1,\ldots,x_j+\epsilon,\ldots,x_n)$ for some $\epsilon>0$.

Let $(P_1,\ldots,P_m) := LS(\mathbf{x},m)$
and $(P_1',\ldots,P_m') := LS(\mathbf{x'},m)$.
Then, the set $\{v(P_i') ~|~ i\in[m]\}$
dominates the set $\{v(P_i) ~|~ i\in[m]\}$.
\end{lemma}

\begin{proof}
The proof is by induction on the number of inputs $n$. For $n=1$ the claim is trivial. 
Suppose the claim is true for $n-1$ inputs. Now there are some $n$ inputs. Consider two cases.

\textbf{Case \#1:}
The value that becomes larger is $x_n$.
Then in the second run, LS partitions the first $n-1$ inputs in exactly the same way as in the first run.
Therefore, when the new input $x_n+\epsilon$ arrives, the situation in the second run is exactly the same as in the first run, and algorithm LS inserts this new input into the same bin into which $x_n$ was inserted in the first run. Therefore, the claimed domination holds with the identity permutation.

\textbf{Case \#2:}
The value that becomes larger is one of $x_1,\ldots, x_{n-1}$.
By the induction assumption, the claim holds for the instance $x_1,\ldots, x_{n-1}$. 
Denote by $v^{k}$ the function that maps a part into its value after iteration $k$ of LS. By the induction assumption,
the set  $\{v^{n-1}(P_i') ~|~ i\in[m]\}$
dominates the set $\{v^{n-1}(P_i) ~|~ i\in[m]\}$.
Let $\pi$ be the permutation such that, for all $i\in [m]$,
$v^{n-1}(P'_{\pi(i)})\geq v^{n-1}(P_i)$.

It remains to show that domination holds after $x_n$ is added. That is, that
the set  $\{v^{n}(P_i') ~|~ i\in[m]\}$
dominates the set $\{v^{n}(P_i) ~|~ i\in[m]\}$. 
Suppose that, in the first run, LS inserted $x_n$ into $P_i$ (for some $i\in[m]$).
Consider two subcases.

\emph{Case 2.1:} In the second run, LS inserts $x_n$ into $P'_{\pi(i)}$.
Then, 
\begin{align*}
v^{n}(P'_{\pi(i)}) &= 
v^{n-1}(P'_{\pi(i)}) + x_n
\\
&\geq 
v^{n-1}(P_{i}) + x_n && \text{by domination}
\\
&=
v^{n}(P_{i}),
\end{align*}
and the other $m-1$ bins do not change in iteration $n$, so domination holds with the same permutation $\pi$.

\emph{Case 2.2:} In the second run, LS inserts $x_n$ into  $P'_{\pi(j)}$ for some $j\neq i$.
By definition of LS, this means that $v^{n-1}(P'_{\pi(j)})$ is a smallest element in the vector $P'$. In particular,
\begin{align*}
v^{n-1}(P'_{\pi(i)})
\geq
v^{n-1}(P'_{\pi(j)}).
\end{align*}
By domination,
\begin{align*}
v^{n-1}(P'_{\pi(j)}) \geq v^{n-1}(P_j).
\end{align*}
Since in the first run, LS inserted $x_n$ into $P_i$, its sum was smallest then, that is,
\begin{align*}
v^{n-1}(P_j) \geq v^{n-1}(P_i).
\end{align*}
Combining these three inequalities gives 
\begin{align*}
v^{n-1}(P'_{\pi(j)})
&\geq
v^{n-1}(P_i),
\\
v^{n-1}(P'_{\pi(i)})
&\geq
v^{n-1}(P_j).
\end{align*}
Therefore,
\begin{align*}
v^{n}(P'_{\pi(j)}) &= 
v^{n-1}(P'_{\pi(j)}) + x_n
\\
&\geq 
v^{n-1}(P_{i}) + x_n
\\
&=
v^{n}(P_{i}),
\\
v^{n}(P'_{\pi(i)}) &= 
v^{n-1}(P'_{\pi(i)})
\\
&\geq 
v^{n-1}(P_{j})
\\
&=
v^{n}(P_{j}).
\end{align*}
So domination holds with a permutation $\pi'$ which is derived from $\pi$ by swapping $\pi(i)$ with $\pi(j)$.
\end{proof}

\subsection{Longest Processing Time}
Theorem \ref{thm:mon-lpt} follows from the following.
\begin{lemma}
\label{lem:mon-lpt}
Let $m\geq 2$ be a fixed integer,  $\mathbf{x} := (x_1,\ldots,x_j,\ldots,x_n)$ a sequence of positive numbers,
and $\mathbf{x'} := (x_1,\ldots,x_j+\epsilon,\ldots,x_n)$ for some $\epsilon>0$.

Let $(P_1,\ldots,P_m) := LPT(\mathbf{x},m)$
and $(P_1',\ldots,P_m') := LPT(\mathbf{x'},m)$.
Then, the set $\{v(P_i') ~|~ i\in[m]\}$
dominates the set $\{v(P_i) ~|~ i\in[m]\}$.
\end{lemma}
\begin{proof}
If the increase in $x_j$ does not change the relative order between the inputs, then LPT runs exactly like LS on the ordered input sequence, so the claim is implied by Lemma \ref{lem:mon-ls}.

Otherwise, the increase in $x_j$ can be decomposed into several smaller increases, where in each increase, it becomes equal to the next-higher input, and in the final increase, it attains its final larger value.
The claim holds for each sub-increase by Lemma \ref{lem:mon-ls}. Since domination is a transitive relation, it holds for the grand increase too.
\end{proof}

\subsection{MultiFit}
\citet{coffman1978application} developed another algorithm for number-partitioning.
It uses as a subroutine an algorithm called \emph{First Fit Decreasing (FFD)}, originally developed for the bin-packing problem. FFD is parametrized by a positive number $c>0$ called the \emph{bin-capacity}.
\begin{enumerate}
\item Order the inputs from large to small.
\item Open a new empty bin, $B$.
\item For each number $x_i$ in the sequence:
\begin{itemize}
\item If the number can be added into $B$ such that the sum remains at most $c$ --- insert $x_i$ into $B$.
\item Else, skip $x_i$.
\end{itemize}
\item If some inputs remain, return to step 2.
\end{enumerate}
MultiFit uses binary search to find the smallest $c$ such that 
FFD manages to pack all input numbers into at most $m$ bins. Then, it returns the resulting $m$-partition.

\citet{huang2021algorithmic}[Examples 5.1, 5.2] have already shown an example in which MultiFit is not value-monotone. For completeness, a different example is presented below.

\begin{example}[MultiFit algorithm anomalies; based on \citet{coffman1978application}]
Suppose $m=3$ and the input sequence is:
\begin{align*}
44, 24, 24, 22, 21, 17, 8, 8, 6, 6.
\end{align*}
With capacity $c=60$, FFD packs all inputs into 3 bins:
\begin{itemize}
\item $44 + 8 + 8 = 60$
\item $24 + 24+ 6 + 6 = 60$
\item $22 + 21 + 17 = 60$
\end{itemize}
Clearly, packing into 3 bins is impossible with any smaller capacity, so MultiFit returns the above partition, in which the largest sum is $60$.

Now, suppose input \#6 decreases from 17 to 16, so the input sequence is:
\begin{align*}
44, 24, 24, 22, 21, 16, 8, 8, 6, 6.
\end{align*}
Now, with capacity $c=60$, FFD needs 4 bins:
\begin{itemize}
\item $44 + 16 = 60$
\item $24 + 24+ 8 = 56$
\item $22 + 21 + 8 + 6 = 57$
\item $6$
\end{itemize}
so MultiFit must increase the capacity. The smallest capacity for which FFD packs all inputs into 3 bins is 62:
\begin{itemize}
\item $44 + 16 = 60$
\item $24 + 24+ 8+6 = 62$
\item $22 + 21 + 8 + 6 = 57$.
\end{itemize}
Thus, the largest sum increases, so MultiFit is not value-monotone.
\qed
\end{example}

\section{Related Work: Speed Monotonicity}
\label{sec:related}
There is another monotonicity property, that is related to algorithms for \emph{uniform-machines scheduling}. In this setting, the $m$ machines may have different \emph{speeds}. When a process is scheduled on machine $i$, its actual run-time is its basic processing-time divided by the machine speed.
\begin{definition}
An algorithm $A$ for uniform-machines scheduling is called \emph{speed-monotone}
if, for any $i\in[m]$, if the speed of machine $i$ increases while all other inputs remain the same, then the total processing-time allocated to $i$ weakly increases.
\end{definition}
Speed-monotonicity is particularly important when the speeds are private information of the machine owners, and it is required to incentivize the owners to reveal the true speeds. 
Such truthful algorithms were discussed by 
\citet{archer2001truthful,auletta2004deterministic,ambrosio2004deterministic,andelman2005truthful,kovacs2005fast}. It is clearly different than the value-monotonicity property studied in the present paper.

\section{Future Work}
The results in this note raise various questions for future work.

\paragraph{1.} What other number-partitioning algorithms are value-monotone? Particularly interesting are the \emph{differencing method} of \citet{karmarkar1982differencing} and its variants, such as \emph{Largest Differencing Method (LDM)}, \emph{Pairwise Differencing Method (PDM)} and \emph{Balanced LDM} \citep{yakir1996differencing}. Are they value-monotone?

\paragraph{2.} Do LS and LPT remain value-monotone for more complex problems? Particularly, are they value-monotone for uniform-machines scheduling? Are they value-monotone for number-partitioning with cardinality constraints on the subsets, such as the \emph{balanced number partitioning}  problem \citep{coffman1984note,tsai1992asymptotic} or the \emph{$k$-partitioning}  problem \citep{babel1998thek,he2003kappa}?

\paragraph{3.} Define a number-partitioning algorithm $A$ as \emph{subset-count-monotone} if increasing the number of allowed subsets ($m$) weakly \emph{decreases} the largest and the smallest subset-sums. 
List Scheduling is not subset-count-monotone when there are dependencies between jobs; this can be shown in Example \ref{exm:ls}.
If the number of machines increases from $3$ to $4$,
then the largest sum \emph{increases} from $121$ to $150$.
What if there are no dependencies --- are LS, LPT and other number-partitioning algorithms subset-count-monotone?

\bibliographystyle{apalike}
\bibliography{../erelsegal-halevi}

\begin{thebibliography}{}

\bibitem[Ambrosio and Auletta, 2004]{ambrosio2004deterministic}
Ambrosio, P. and Auletta, V. (2004).
\newblock Deterministic monotone algorithms for scheduling on related machines.
\newblock In {\em International Workshop on Approximation and Online
  Algorithms}, pages 267--280. Springer.

\bibitem[Andelman et~al., 2005]{andelman2005truthful}
Andelman, N., Azar, Y., and Sorani, M. (2005).
\newblock Truthful approximation mechanisms for scheduling selfish related
  machines.
\newblock In {\em Annual Symposium on Theoretical Aspects of Computer Science},
  pages 69--82. Springer.

\bibitem[Archer and Tardos, 2001]{archer2001truthful}
Archer, A. and Tardos, {\'E}. (2001).
\newblock Truthful mechanisms for one-parameter agents.
\newblock In {\em Proceedings 42nd IEEE Symposium on Foundations of Computer
  Science}, pages 482--491. IEEE.

\bibitem[Auletta et~al., 2004]{auletta2004deterministic}
Auletta, V., De~Prisco, R., Penna, P., and Persiano, G. (2004).
\newblock Deterministic truthful approximation mechanisms for scheduling
  related machines.
\newblock In {\em Annual Symposium on Theoretical Aspects of Computer Science},
  pages 608--619. Springer.

\bibitem[Babel et~al., 1998]{babel1998thek}
Babel, L., Kellerer, H., and Kotov, V. (1998).
\newblock Thek-partitioning problem.
\newblock {\em Mathematical Methods of Operations Research}, 47(1):59--82.

\bibitem[Coffman et~al., 1978]{coffman1978application}
Coffman, Jr, E.~G., Garey, M.~R., and Johnson, D.~S. (1978).
\newblock An application of bin-packing to multiprocessor scheduling.
\newblock {\em SIAM Journal on Computing}, 7(1):1--17.

\bibitem[Coffman~Jr et~al., 1984]{coffman1984note}
Coffman~Jr, E.~G., Frederickson, G.~N., and Lueker, G.~S. (1984).
\newblock A note on expected makespans for largest-first sequences of
  independent tasks on two processors.
\newblock {\em Mathematics of Operations Research}, 9(2):260--266.

\bibitem[Graham, 1966]{graham1966bounds}
Graham, R.~L. (1966).
\newblock Bounds for certain multiprocessing anomalies.
\newblock {\em Bell system technical journal}, 45(9):1563--1581.

\bibitem[Graham, 1969]{graham1969bounds}
Graham, R.~L. (1969).
\newblock Bounds on multiprocessing timing anomalies.
\newblock {\em SIAM Journal on Applied Mathematics}, 17(2):416--429.

\bibitem[He et~al., 2003]{he2003kappa}
He, Y., Tan, Z., Zhu, J., and Yao, E. (2003).
\newblock $\kappa$-partitioning problems for maximizing the minimum load.
\newblock {\em Computers \& Mathematics with Applications},
  46(10-11):1671--1681.

\bibitem[Huang and Lu, 2021]{huang2021algorithmic}
Huang, X. and Lu, P. (2021).
\newblock An algorithmic framework for approximating maximin share allocation
  of chores.
\newblock In {\em Proceedings of the 22nd ACM Conference on Economics and
  Computation}, pages 630--631.

\bibitem[Karmarkar and Karp, 1982]{karmarkar1982differencing}
Karmarkar, N. and Karp, R.~M. (1982).
\newblock {\em The differencing method of set partitioning}.
\newblock Computer Science Division (EECS), University of California Berkeley.

\bibitem[Kov{\'a}cs, 2005]{kovacs2005fast}
Kov{\'a}cs, A. (2005).
\newblock Fast monotone 3-approximation algorithm for scheduling related
  machines.
\newblock In {\em European Symposium on Algorithms}, pages 616--627. Springer.

\bibitem[Tsai, 1992]{tsai1992asymptotic}
Tsai, L.-H. (1992).
\newblock Asymptotic analysis of an algorithm for balanced parallel processor
  scheduling.
\newblock {\em SIAM Journal on Computing}, 21(1):59--64.

\bibitem[Yakir, 1996]{yakir1996differencing}
Yakir, B. (1996).
\newblock The differencing algorithm ldm for partitioning: a proof of a
  conjecture of karmarkar and karp.
\newblock {\em Mathematics of Operations Research}, 21(1):85--99.

\end{thebibliography}

\end{document}